\documentclass[letterpaper]{article}
  \usepackage{aaai18}
  \usepackage{amsmath}
  \usepackage{times}
  \usepackage{helvet}
  \usepackage{courier}
  \usepackage{url}
  \usepackage{graphicx}
  \usepackage{color}
  \usepackage{algorithm}
  \usepackage{bbm}
  \usepackage{bm}
  \usepackage{amsthm}
  \newtheorem{proposition}{Proposition}
  
  \newtheorem{lemma}{Lemma}
  \theoremstyle{definition}
  
  \usepackage[noend]{algpseudocode}
  \usepackage{pgfplots}
  \usetikzlibrary{patterns}
  \usepackage{subcaption}
  \DeclareMathOperator*{\argmax}{arg\,max}
\pgfplotsset{width=6cm,compat=1.9}
\title{Activating the ``Breakfast Club": Modeling Influence Spread in Natural-World Social Networks}
\author{Lily Hu \textsuperscript{1}, Bryan Wilder \textsuperscript{2},
Amulya Yadav \textsuperscript{2},
Eric Rice \textsuperscript{2},
Milind Tambe \textsuperscript{2}\\
\textsuperscript{1}Harvard University, Cambridge, MA, 02138, USA\\
\{lilyhu\}@g.harvard.edu\\
\textsuperscript{2}University of Southern California, Los Angeles, CA, 90089, USA\\
\{bwilder, amulyaya, ericr, tambe\}@usc.edu
}
\title{Activating the ``Breakfast Club": Modeling Influence Spread in Natural-World Social Networks}

\begin{document}
  \maketitle
  \begin{abstract}
  While reigning models of diffusion have privileged the structure of a given social network as the key to informational exchange, real human interactions do not appear to take place on a single graph of connections. Using data collected from a pilot study of the spread of HIV awareness in social networks of homeless youth, we show that health information \textit{did not} diffuse in the field according to the processes outlined by dominant models. Since physical network diffusion scenarios often diverge from their more well-studied counterparts on digital networks, we propose an alternative Activation Jump Model (AJM) that describes information diffusion on physical networks from a \textit{multi-agent team} perspective. Our model exhibits two main differentiating features from leading cascade and threshold models of influence spread: 1) The structural composition of a seed set team impacts each individual node's influencing behavior, and 2) an influencing node may spread information to non-neighbors. We show that the AJM significantly outperforms existing models in its fit to the observed node-level influence data on the youth networks. We then prove theoretical results, showing that the AJM exhibits many well-behaved properties shared by dominant models. Our results suggest that the AJM presents a flexible and more accurate model of network diffusion that may better inform influence maximization in the field.

\end{abstract}
\section{Introduction}

Research in influence maximization lies at the crossroads of the \textit{who}, \textit{what}, and \textit{when} of information diffusion. In their seminal paper, Kempe, Kleinberg, and Tardos (\citeyear{kempe2003maximizing}) formalized the problem by imagining influencing agents as seed nodes in a network initialized to propagate information first to neighbors and eventually throughout the network as influence spreads. Along with theoretical work in the field, the wide availability of large-scale digital data has positioned internet networks---web link traces, email communication, social media platforms---at the center of the discussion of the influence maximization problem. Even so, empirical work that compares predictions by the dominant Independent Cascade (ICM) and Linear Threshold (LTM) models with realized diffusion at the \textit{node level} is limited. Prevailing methodologies for estimating diffusion model parameters often achieve low accuracy in replicating observed behavior even when applied to well-defined online networks with temporal information flow data \cite{goyal2011data,wang2013learning}. 


Although the ICM and LTM were originally formulated to describe social influence in natural environments \cite{granovetter1978threshold,schelling1978micromotives}, there is a dearth of high-quality data that support the theories on networks in these physical settings. Moreover, multiple challenges exacerbate these deficiencies in physical settings. First, both the ICM and LTM assume that the topology of social ties is identical with the mesh of connective channels through which information spreads. While on many social media platforms a user's social network delimits her space of communication, in the natural world, an individual's total space of social navigation dwarfs the space of those she calls her ``friends,'' and there exists a multiplicity of information avenues within the network that do not coincide with one's social ties \cite{shafie2015multigraph}.
Further, previous work has shown that models with strong assumptions about a particular constructed graph topology are more prone to error and inaccuracy in their predictions of information spread \cite{butts2003network,krackhardt1999whether}.

In this paper, we analyze information diffusion data from the first empirical study of influence maximization in the \textit{physical world} \cite{yadav2017influence}, which tracked 173 individuals across 3 distinct networks over a multi-year time period. We delved into this significant corpus of natural world network data to investigate influence spread at an \textit{individual node level} rather than a network-wide volume level. We found that information \textit{did not} diffuse from seed nodes to the greater network according to processes suggested by the ICM or LTM. Most strikingly, we found that across all three networks in the study, \textit{$\it{50\%}$ of informed nodes lacked any path to a seed node} and moreover, a node's degree of connectivity---both generally and specifically to seed nodes---exhibited \textit{no correlation with likelihood of becoming informed}. These results directly contradict predictions put forth by the ICM and LTM and call into question the suitability of these leading models of diffusion for approximating information spread on physical networks.


These negative results against the ICM and LTM can be generalized to apply to other graph-based models of diffusion: The high proportion of informed yet isolated nodes cannot be explained by models that rely on edge-based propagation as the sole avenue for influence spread \cite{degroot1974reaching}. On the other hand, more flexible approaches such as Hawkes contagion processes \cite{zhou2013learning}, which allow for influence to spread over greater social distances, cannot explain the differing spread outcomes across the networks. In general, appealing only to non-graph features of the diffusion scenarios cannot reconcile the divergences in observed information diffusion in the three different Yadav et al. studies. In response to the shortcomings of existing models, we develop a new model of \textit{targeted peer-to-peer information spread on natural networks} that does not rely on strong tie assumptions and instead incorporates an understanding of influencing as a  ``team'' behavior.

Our model features two distinguishing characteristics that are aligned with real-world information diffusion in physical spaces: 1) Nodes exchange information beyond their immediate social ties, and 2) Seed nodes act as a \textit{multi-agent team} to spread information, where their overall influencing efficacy is a function of both individual and team attributes. In the proposed \textit{Activation Jump Model} (AJM), team-based influence spread in a network is driven by activating the ``Breakfast Club," where individuals from different social contexts band together for a common cause and form a united team for information diffusion. These features confer a flexibility to our forecasts of information flow and allow our model to \textit{achieve a 60\% to 110\% improvement over the best ICM and LTM predictions in its predictions of which nodes will be influenced}.

We also point to a methodological pitfall of research in influence maximization that focuses solely on achieving a particular level of information diffusion within a network. Namely, matching magnitude of influence spread under simulations to observed influence spread is insufficient evidence for determining the underlying diffusion process. We show that optimal seeding under one model achieves \textit{near-optimal} ($> 90\%$) influence spread under another diffusion process on three natural-world networks. In fact, \textit{any} magnitude of influence spread can be explained by varying ICM and LTM parameters, pointing to a fundamental ambiguity in identifying the true diffusion process based on this metric alone. 

Even when seeding strategies achieve high levels of influence spread, leading models' failures to predict node-level influence can limit their applicability. In domains of sustainability, network interventions can generate knowledge or promote behavioral changes within a community. These programs typically identify individuals and groups that may especially benefit from the intervention. For example, school network-based suicide prevention programs aim to increase general awareness about signs of suicidal behavior but especially seek to reach high-risk adolescents and their social circles \cite{kalafat1994evaluation}. Similarly, peer-led HIV prevention programs akin the to the fieldwork by Yadav et al. hope to reach a diverse set of individuals but especially those who participate in risky behaviors \cite{broadhead1998harnessing}. As such, most social interventions have the dual purpose of maximizing influence coverage while also targeting vulnerable individuals. With its superior performance in predicting node-level influence, the AJM may serve as a more desirable framework to guide these intervention strategies in the field.

\section{Network Intervention Data Analysis}

\subsubsection{Pilot Study Procedure}
Yadav et al.'s long-standing collaboration with homeless youth service providers in a large urban area sought to improve peer-led heatlh interventions by leveraging research in influence maximization (\citeyear{yadav2017influence}). To this end, they conducted a series of head-to-head comparison studies of seeding strategies to select cohorts of \textit{Peer Leaders} among the youth that would be trained for the task of HIV awareness diffusion in their communities. Three studies took place on three distinct social networks of homeless youth. Each study recruited youth and gathered social network data using online contacts, field observations, and surveys. A different seeding strategy was then deployed on each of the generated networks: In two of the pilot studies, Peer Leaders were chosen via two algorithmic agents for influence maximization, HEALER and DOSIM, which were designed to optimize network-based intervention strategies for health providers. The third network was seeded via degree centrality (DC), the most commonly-used heuristic in network interventions \cite{valente2012network}, such that the most popular youth were chosen as Peer Leaders. 

Each network's Peer Leaders underwent an intensive training course led by pilot study staff that served to both instruct the youth in spreading information about HIV to their peers as well as bind the members together in their shared roles as health ambassadors. After Peer Leaders were sent out into the field, youth were asked in 1-month and 3-month follow-up surveys about whether they had received information about HIV from a Peer Leader. These responses revealed the extent to which information had spread from seed nodes to the greater network. 
The post-intervention results revealed that the HEALER and DOSIM seeding strategies resulted in greater informational spread compared to DC, with ${\sim} 74\%$ and ${\sim} 72\%$ respectively of non-Peer Leaders reporting having received information about HIV in the 3-month survey compared to ${\sim} 35\%$ in the control study. Since both HEALER and DOSIM solved the influence maximization problem by assuming a model of information spread based on a generalization of the Independent Cascade, the success initially seemed to validate the model as an accurate approximation of information spread in the physical world. 

\begin{table}[h]
\centering
\caption{Fraction of nodes in each network that were informed, sorted by connectivity status to Peer Leaders. Denominator gives total number of nodes of that type; numerator gives number of those that are informed. Direct nodes have an edge to a PL; indirect nodes are connected via intervening neighbor(s); isolated nodes lack a path to any PL.}
\label{table:overview}
\scalebox{0.9}{
\begin{tabular}{ccccc|c}
\textbf{Network} & $\bm{n}$ & \textbf{Direct} & \textbf{Indirect} & \textbf{Isolated} & \begin{tabular}[c]{@{}c@{}}\textbf{Proportion}\\ \textbf{Informed}\end{tabular} \\ \hline
HEALER & 34 & 15/21 & 4/7 & 6/6 & 25/34 \\ \hline
DOSIM & 25 & 5/6 & 5/10 & 8/9 & 18/25 \\ \hline
DC & 26 & 5/12 & 1/4 & 3/10 & 9/26 \\ \hline
\end{tabular}
}
\end{table}

\subsubsection{Node-level analysis of information diffusion}
However, the empirically observed \textit{node-level} patterns of information spread in the three networks wildly diverged from Independent Cascade and Linear Threshold predictions. Table \ref{table:overview} gives an overview of the connectivity of nodes that reported receiving information about HIV from a Peer Leader in the 3-month follow-up survey.
In each of the three networks, nodes lacking a path to any seed Peer Leaders---denoted as ``isolated" in the table--- represented a high proportion of all nodes that were informed.
Notably, in both the HEALER and DOSIM interventions, isolated youth were informed at a rate \textit{higher} than even those youth who were directly connected to one or more Peer Leaders, with 100\% (6/6 in HEALER) and 89\% (8/9 in DOSIM) informed compared to ${\sim}71\%$ (15/21) and ${\sim}83\%$ (5/6). In the DC network, the effect is less pronounced, though isolated nodes were still informed at a rate comparable to the general non-Peer Leader population (30\% compared to ${\sim}35\%)$. Nonetheless within the context of the ICM and LTM, such nodes have a $0$ probability of receiving information. Thus, these results immediately challenge the claim that existing ties are the dominant avenues of informational exchange and also call into question the premise that information radiates out from seed nodes first to neighbors and then to the rest of the network.

In order to more finely assess the effect that a node's connectivity had on its likelihood of receiving HIV information, we calculated Pearson correlation coefficients between two degree measures and a node's final information status. Our results in Table \ref{table:corr} show that all such correlations are not significantly different from no correlation, thus indicating that connectivity has \textit{no bearing} on likelihood of being influenced. This stands in contrast to prevailing models, in which a node's edges represent its ``opportunities" to receive information, and thus both Peer Leader degree---the number of ties a node has to Peer Leaders---and total degree should be strictly positively correlated with becoming informed.

\begin{table}[ht]
\centering
\caption{Pearson correlation coefficients $r$ between nodes' influence statuses and their Peer Leader (PL) and total degree. PL degree counts edges to PLs; total degree counts edges to all nodes. Positive (negative) $r$ implies a positive (negative) relationship between reporting hearing HIV information and degree. 95\% confidence intervals are also given.}
\label{table:corr}
\scalebox{0.9}{
\begin{tabular}{cccc}
\textbf{Network} & \multicolumn{1}{c}{$\bm{n}$} &\textbf{PL Degree} & \textbf{Total Degree}\\ \hline
HEALER & 34 & \begin{tabular}[c]{@{}c@{}}-0.0685\\ (-0.397, 0.276)\end{tabular} & \begin{tabular}[c]{@{}c@{}}-0.1867\\ (-0.494, 0.162)\end{tabular}\\ \hline
DOSIM & 25 & \begin{tabular}[c]{@{}c@{}}0.1418\\ (-0.268, 0.508)\end{tabular} & \begin{tabular}[c]{@{}c@{}}-0.0290\\ (-0.419, 0.370)\end{tabular} \\ \hline
DC & 26 & \begin{tabular}[c]{@{}c@{}}0.1547\\ (-0.247, 0.511)\end{tabular} & \begin{tabular}[c]{@{}c@{}}0.1304\\ (-0.271, 0.493)\end{tabular}\\ \hline
\end{tabular}
}
\end{table}

The lower contact rate of directly connected nodes and the lack of positive correlation between degree and influence status are even more dissonant with edge-based models of propagation when considered alongside the high levels of influence spread achieved in the studies. In the HEALER network, information was successfully transmitted to ${\sim}74\%$ of all non-Peer Leaders, corresponding to a most likely propagation probability of $p \approx 0.84$. Such a high propagation probability further suggests that Peer Leader-neighboring nodes should be even more heavily favored to receive information, with simulations predicting that \textit{nearly all} (${\sim}99\%$) would become informed, whereas in reality, only ${\sim}71\%$ of these nodes received information. Simulations on the HEALER network with this $p$ value produce correlations of $0.489$ and $0.598$ between degree and likelihood of being informed (PL and total respectively), indicating a moderate to strong positive relationship compared to the actual values of $-0.0685$ and $-0.1867$, which indicate negative to no relationship between degree and influence status. In the DOSIM study, the graph topology itself, with 5 connected components in addition to 7 nodes of degree 0, restricts information spread under the ICM and LTM to maximally reach $68\%$ of all non-Peer Leaders. Even under perfect information propagation, as long as nodes are only able to influence neighbors, simulations under-predict the observed information spread.


While the presence of a single phenomenon such as the activation of a small proportion of isolated nodes could represent mere aberrations of data, these multiple contradictions with prevailing models indicate that the results cannot be dismissed as simply anomalous. Given the unique challenges and complexities of information diffusion on physical networks, we must accept that the data's divergence from predictions by models that have been largely validated only on digital networks is one step in the uncovering and understanding of a qualitatively different influence process. We thus conclude that there is no evidence that a cascade or threshold-like process of information diffusion produced the observed data and move toward developing a new model of influence dynamics on real-world physical networks.


\section{Proposed Model} 
In this section, we introduce a new model of information spread for this class of peer-to-peer diffusion phenomena. 
\subsection{Activation Jump Model}
Beginning with the premise that instances of social exchange are not limited to nodes that share a tie, our model of diffusion does not constrain information flow to the edges in a network. In the \textit{Activation Jump Model} (AJM), influencing agents may leave their immediate social neighborhood to contact and propagate information to other nodes. This action of contacting nodes beyond one's first-order ties is signified as a ``jump." We recognize the heterogeneity of active nodes' social dispositions by differentially modeling each influencer's jump behavior. A seed node's jump activity has two main components: 1) \textit{activation level}, a measure of \emph{how many} other nodes it will attempt to influence, and 2) \textit{landing distribution}, a probability distribution that expresses to \emph{which} inactive nodes it will jump. Together, these two features describe \textit{how often} and \textit{to whom} an influencing agent contacts as she navigates the network to spread information.

Thus the AJM comprises two stages: First, each seed node determines its activation level, giving the number of other nodes to which it will jump. Second, the seed set is deployed in the network, and the social influence process unfolds in time. When a given seed jumps at time $t$, it selects from its landing distribution a target node uninformed at time $t$, modeling the process by which seed nodes seek nodes to inform. Influence is then successfully propagated with probability $p$.

The AJM takes a multi-agent systems approach to the influence maximization problem by constructing a model of node activation that is a function of both individual and ``team" attributes. In contrast to prior models, the seed set is not a collection of independent influencers, rather nodes exhibit behavioral dependencies wherein group dynamics either contribute to or detract from aggregate activation levels.


\subsection{Model Formalization}
While in this paper and all our results, we use a form of the AJM tied to the graph's structural properties, we first discuss the general form of the model to show that it can accommodate a broad range of properties and then discuss our specific form. We return to the generalization in the Discussion.

Consider a team of seed nodes, $S$, tasked with information diffusion on a network $G = (V, E)$. Each seed node draws its activation level, giving the number of jumps it will make, from a distribution that is a function of both the node's \textit{individual} attributes as well as the seed set's \textit{team} attributes. Formally, let $\Delta^{Z}$ be the set of distributions over integers $Z \ge 0$. Each node $v \in V$ is associated with a function $f_v: 2^{|V|} \rightarrow \Delta^Z$ that maps the set of seed nodes to a distribution $A(v,S)$ over discrete activation levels. $A(v,S)$ is a parameterized distribution (e.g., geometric) with mean $\mu_v = h(S)[a^Tx_v]$ where $x_v$ is the node's attribute vector with coefficients $a^T$. Together $a^Tx_v$ represents the particular node's maximum activation level, which is modulated by the team activation level term given by $h(S) \le 1$, a function of the structural positions of nodes in $S$ that captures discomplementarities among team members. Figure \ref{fig:AJM} illustrates this dual---individual and team---composition of a seed node's activation level distribution. 

\begin{figure}[ht]
\includegraphics[scale=0.25]{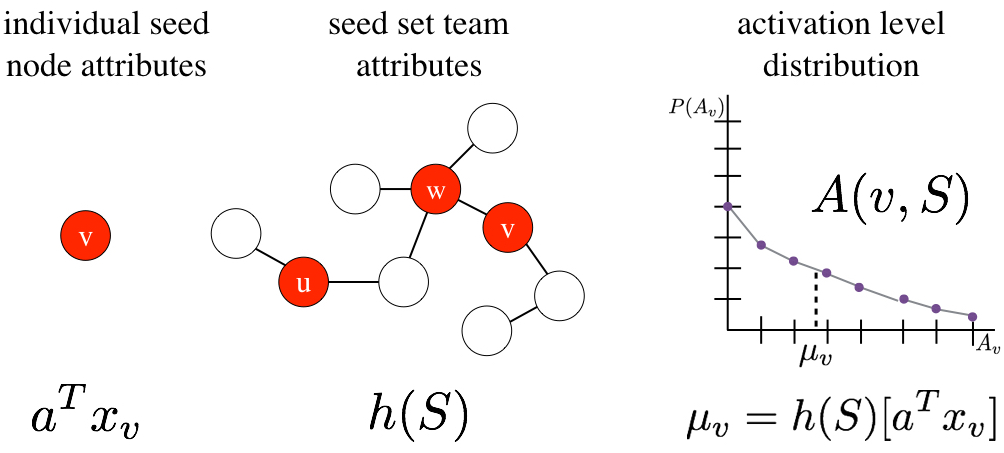}
\caption{In the Activation Jump Model, a seed node $v$ is associated with an activation level distribution that is a function of individual node as well as seed set $S$ attributes.}
\label{fig:AJM}
\end{figure}

Each node $v$ is also associated with a \textit{landing distribution}, $L_{v,T}$, giving the probabilities with which $v$ jumps to a set of potential target nodes $T$. The landing probability is a function of the attributes of the influencing seed $v$ and particular targeted node $u$. Based on these qualities, the node pair is assigned a score $\phi(v, u) \ge 0$, and $L_{v, T}(u) = \Gamma \phi(v, u)$ where $\Gamma$ is a normalization factor such that $\sum_{u \in T} L_{v,T}(u) = 1$. 

We now instantiate the AJM in a specific form that features the concept of ``structural diversity," which highlights groups with members who participate in multiple distinct social contexts. The team thus acts to \textit{unite} otherwise disparate nodes, producing the ``Breakfast Club" effect, which has been shown to be a key determinant of diffusion in networks \cite{ugander2012structural}. Thus we formulate the function 
\begin{align}
h(S) = 1 - \frac{1}{A}\sum_{(u,v) \in E} \mathbbm{1}\left[u, v \in S\right] 
\label{eq:h}
\end{align}
where $A>0$ is a constant equal to seed set size $|S|$. $h$ illustrates the negative effect of social homogeneity in the form of between-seed-node edges on a team's effectiveness. Each pair of connected seed nodes entails a loss of $\frac{1}{|S|}$ of the team's effectiveness. Barring negative influence, $h'(S) = \max(h(S), 0)$ without loss of generality. To reflect the correlation between degree and propensity towards sociality and thus activation, we parsimoniously set $a^Tx_v = deg(v)$. 

The landing distribution score for seed $v$ and target $u$ is given by $\phi(v,u) = \frac{1}{d(v,u)}$, where $d(v,u)$ is the path-length distance between the two nodes. When $d(v,u) = 0$, we set $\phi(v,u) = \epsilon$, where $\epsilon > 0$ is a small constant. As in the ICM, the propagation probability $p$ is able to be varied. 


Information diffusion thus occurs in two stages. First, during the \textit{activation stage} each node in the seed set $v \in S$ is initialized by drawing an activation level $A_v$ from its distribution $A(v,S)$. Then, the \textit{jump stage} unfolds over the time interval $[0, 1]$. Each seed node $v$ draws a series of jump times $t^v_1...t^v_{A_v}$ from the uniform distribution over $[0,1]$. At each jump time $t^v_i$, $v$ jumps to an uninfluenced target node $u$ drawn from $L_{v, T}$ where $T$ the set of uninfluenced nodes at $t^v_i$. Finally, $u$ is successfully influenced with probability $p$. 




\subsection{Model Discussion}
The Activation Jump Model's incorporation of seed set team dynamics follows a line of multi-agent systems research which demonstrates the importance of careful team formation when agents must collaborate to achieve a goal \cite{gaston2003team,gaston2004adapting,liemhetcharat2012modeling,agmon2012leading}. In particular, previous work has focused on the importance of creating a \emph{diverse} team \cite{balch2000hierarchic,hong2004groups,marcolino2013multi}. In the AJM, we computationalize this concept by using the group effectiveness function $h$ to model network structural diversity by penalizing seed sets with many within-team edges. Thus $h \in [0, 1]$ is decreasing in the level of connectivity among seed nodes, and influencing nodes are more active when they occupy distinct neighborhoods of the network rather than when the team is socially homogeneous. 


It is important to note that under this model set-up, a node's marginal effect on the aggregate activation level of a seed set is not guaranteed to be positive. There may exist a node $w \in V$ such that $\sum_{v \in S} h(S)[a^Tx_v] > \sum_{v \in S \cup w} h(S \cup w)[a^T x_v]$, with the effect that the influence function $f(\cdot)$, giving the expected number of influenced nodes, is non-monotone. Although this is a significant departure from the ICM and LTM, we argue that non-monotonicity is a realistic feature of team-based influence spread, since a new seed node may interfere with team dynamics, resulting in a deleterious effect that outweighs its positive individual contribution. This balance between the quantity and quality of members in a seed set is an important consideration in team formation in the real world. As a result, the influence maximization problem under the AJM requires examination of not only a node's individual attributes but also its effect on the group composition of nodes already in the seed set. 



\subsection{Model Validation on Post-Intervention Data}
We evaluate the performance of the Activation Jump Model by comparing its predictions to Yadav et al.'s dataset of HIV awareness spread on three distinct social networks of homeless youth. Standard experiments of diffusion models compare the magnitude of total influence spread under simulations to that observed empirically in order to assess model accuracy. Here, we perform a finer-grained analysis by evaluating and comparing AJM, ICM, and LTM predictions of node-level influence. We treat each model as a binary classifier that outputs the predicted probability of each node becoming influenced. Each model is then evaluated according to its AUROC, a standard measure of classification accuracy.


\noindent\textbf{Parameter settings: }Physical networks present challenges in data collection that limit the ability to view multiple cascades, rendering standard methods of inferring diffusion parameters inoperable. We work under this constraint by fitting the ICM directly to the test data by running simulations under the propagation value that gives its best classification performance and then forcing the AJM to also work under this value. Thus any experimental bias favors the ICM. 

For the AJM, the only parameter we set is the small constant $\phi(v,u) = 0.1$ for the landing distribution score when $d(v,u) = 0$. By contrast, we present the strongest possible version of the ICM for each network, fitting it directly to the test data by selecting the propagation value $p$ that maximized the ICM's AUROC value. We then used this same probability for the AJM. By forcing the AJM to operate under the ICM's optimal parameterization, we ensure that our experiments truly test the AJM's better suitability for modeling the data, rather than a better ability to ``memorize'' the data. 

\begin{figure}
\centering
\begin{tikzpicture}
\begin{axis}[legend entries={AJM, ICM, LTM},
	symbolic x coords = {HEALER, DOSIM, DC},
	xtick=data,
	ylabel=AUROC,
    xlabel=Network,
	enlargelimits=0.15,
    enlarge x limits = 0.22,
	legend style={at={(0.5,-0.32)},
	anchor=north,legend columns=-1},
    nodes near coords,
    nodes near coords align={vertical},
	ybar, area legend
]
\addplot[fill=black] 
	coordinates {(HEALER, 77) (DOSIM, 75) (DC, 61)};
\addplot[pattern=grid] 
	coordinates {(HEALER, 48) (DOSIM, 46) (DC, 29)};
\addplot[pattern=north east lines]
	coordinates{(HEALER, 44)(DOSIM, 34) (DC, 20)};
\end{axis}
\end{tikzpicture}
\caption{Comparison of model classification performances}
\label{fig:ROC-comparison}
\end{figure}
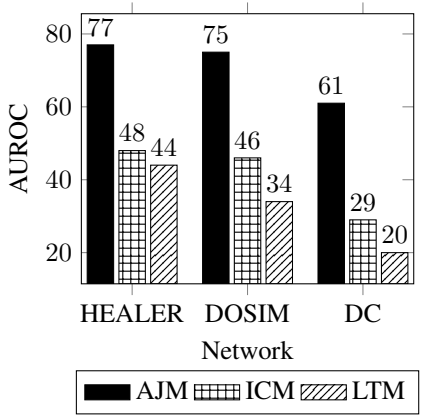

\noindent\textbf{Assessing classification accuracy: }Using selected Peer Leaders in the field experiments as seed nodes, we generated diffusion instances according to the Activation Jump, Independent Cascade, and Linear Threshold models, tracing out Receiver Operating Characteristic (ROC) curves for each set of simulations. This evaluation methodology has been used in previous node-level analyses of information diffusion models \cite{wang2013learning,goyal2010learning} and has been recognized as superior to Precision-Recall curves for the binary classification task \cite{provost1998case}. ROC curves plot a classifier's True Positive Rate (TPR) against its False Positive Rate (FPR) with each point on the curve corresponding to a predictive threshold such that all nodes with a probability of being informed above (below) the threshold are classified as influenced (not influenced). We used the area under the ROC curve (AUROC) to evaluate classification performance \cite{fawcett2006introduction} where an AUROC of $1$ corresponds to a perfect classifier.  

\noindent\textbf{Results: } Each model's AUROC values for the three networks are shown in Figure \ref{fig:ROC-comparison}; the ROC curves for all three models' predictions on the HEALER network are shown in Figure \ref{fig:ROC-HEALER}. The AJM outperforms the ICM and LTM across all networks, with the model achieving accuracies (measured via AUROC) of 77\% and 75\% on the HEALER and DOSIM networks respectively, while the best ICM and LTM issue predictions that, on average, perform worse than a random classifier (20-48\%). For DC, one possible explanation for all three models' lower AUROCs is the overall poor permeation of influence throughout the network, since low base-rates cause the measure to be sensitive to small classification changes. Even so, the AJM is far from a trivial classifier, with an AUROC of 0.61 compared to the ICM and LTM values of 0.29 and 0.20 respectively.
\begin{figure}
\centering
\includegraphics[scale=0.26]{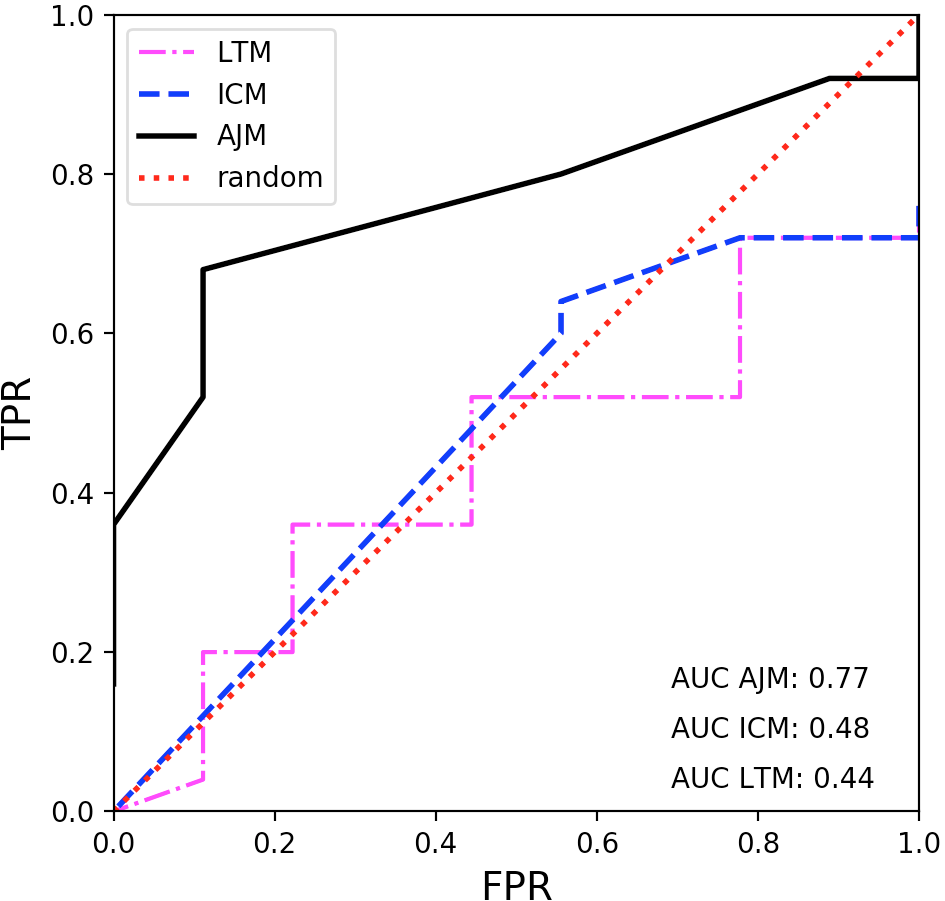}
\caption{Models' ROC curves on HEALER network}
\label{fig:ROC-HEALER}
\end{figure}

\section{The Influence Maximization Problem under the Activation Jump Model}

We now consider the influencing activity of a coordinated multi-agent team under the AJM. Since seed agents do not target nodes that have already been informed, the influence function is captured by the total number of expected jumps, given by $f(S) = h(S)\sum_{v \in S}a^T x_v$, where $h(S)$ follows the form in Equation \ref{eq:h}. We show that under natural conditions, $f$ is a (potentially nonmonotone) submodular function. 




\begin{lemma}
$h$ is monotone-decreasing and submodular. 
\end{lemma}
\begin{proof}
$h$ is monotonically decreasing by inspection. For submodularity, consider the marginal impact of adding a given node $v$ to an existing seed set $S$:

\begin{align*}
h(v|S) = -\frac{1}{A}\sum_{(u, w) \in E} \mathbbm{1} \left[\{u,w\} \not\subseteq S, \{u, w\}\subseteq S \cup \{v\}\right].
\end{align*}

Now consider some $S' \supseteq S$, $v \not\in S'$. The indicator function in the above sum counts edges where one of the two nodes is not contained in $S$, but both \emph{are} contained in $S \cup \{v\}$. If $S'$ extends $S$ but does not contain $v$, then the summation for $h(v|S')$ can only include more nonzero terms than the summation for $h(v|S)$. Since each term is nonpositive, $h(v|S') \leq h(v|S)$. 
\end{proof}

In fact, $h$ being monotone-decreasing and submodular is sufficient for the objective $f$ to also be submodular:

\begin{proposition}
Whenever $h$ is a monotone-decreasing submodular function, $f$ is submodular.
\end{proposition}
\begin{proof}
Consider the marginal gain of adding a node $v$ to a given seed set $S$:

\begin{align*}
f(v|S) = h(S \cup \{v\})a^T x_v + \left[h(S \cup \{v\}) - h(S)\right]\sum_{u \in S}a^T x_u
\end{align*}

We prove that $f$ is submodular by showing that each corresponding term in $f(v | S')$ can only decrease for all $S \subseteq S'$. The first term decreases since $h$ is a submodular function as shown in the lemma. The second term, corresponding to the individual contribution of $v$, also decreases because $h$ is monotonically decreasing. Thus $f(v|S') \le f(v|S)$.
\end{proof}

Having shown that $f$ is submodular, a natural approach to seeding would use the greedy algorithm, giving a $1-\frac{1}{e}$ approximation for the ICM and the LTM \cite{kempe2003maximizing,leskovec2007cost,chen2009efficient}. However, since $f$ is non-monotone, this approach does not apply. Instead, we adopt the stochastic greedy algorithm proposed by Feldman, Harshaw, and Karbasi (\citeyear{feldman2017greed}).
Algorithm 1 runs the normal greedy algorithm (lines 4-6) but only selects from a limited set of nodes $P$. Each node is included in $P$ with probability 1/2. This random removal reduces the chance that the greedy method will prematurely commit to a node that later become problematic due to non-monotonicity. Feldman et al. (\citeyear{feldman2017greed}) show that this algorithm obtains a guaranteed $\frac{1}{4}$-approximation to the optimal value and has excellent empirical performance. Our experiments follow their suggested strategy of running the algorithm several times (we both use 4).
\begin{algorithm}
\caption{StochasticGreedyAJM $(V, f, k)$}\label{alg:greedyAJM}
\begin{algorithmic}[1]
\State initialize $S = \emptyset, P = \emptyset$ 
\For{$v \in V$}
\State with \textit{probability} $\frac{1}{2}, P \gets P \cup v$
\EndFor
\While{$|S| < k$ \text{and} $\exists v \in P$ such that $f(S \cup v) \ge f(S)$}
\State $v = \argmax_{v \in P} f(S \cup v) - f(S)$
\State $S \gets S \cup v$
\EndWhile
\State \textbf{return} $S$
\end{algorithmic}
\end{algorithm}

\subsection{Meta-Analysis of Influence Spread Metrics}
The finding that the ICM is a poor predictor of node-level influence is dissonant with the fact that seeding algorithms based on the ICM have proved effective in the field \cite{yadav2017influence}. After all, how can an algorithm based on an inaccurate model of diffusion manage to nevertheless achieve a high level of influence spread? To address this seeming conflict, we confront a larger question about the prevailing methodology of the influence maximization problem. In this section, we show that appealing solely to the magnitude of influence spread achieved is a fundamentally inconclusive method of determining whether a particular diffusion model underlies an observed instance of spread. This ambiguity is problematic when using the influence maximization framework to inform the seeding strategies of network interventions in sustainability domains. In many such cases, in addition to diffusing information generally, programs seek to target particular individuals or groups, and thus a model's ability to make node-level predictions is a valuable asset. 

\textbf{Magnitude of Influence Spread } Previous research comparing information diffusion model predictions to empirical results has tended to rely on metrics related to volume of spread---such as minimizing RMSE as a function of actual spread or recapitulating cascade sizes---to determine the fidelity of a model to ground truth processes \cite{goyal2011data}. However, one cannot extrapolate \textit{processes} from such coarse-grain \textit{outcomes}. The following experiments use three examples of physical, meso-scale networks: \emph{Homeless}, a network of 142 nodes gathered via interviews with homeless youth, \emph{India}, a household-level network gathered from a rural village in India \cite{banerjee2013diffusion}, and \emph{SBM}, a synthetic network of 200 nodes generated via the Stochastic Block Model, which replicates the community structure found in real social networks.

We evaluate how seed sets selected under one diffusion model perform in an influence maximization task under the other models. Figure \ref{fig:cross-validation} examines the consequences of model misspecification for influence maximization. We set the parameters equally across all networks---0.1 for propagation probabilities in the ICM and AJM and edge weights in the LTM. Each table entry shows the percentage of optimal influence spread obtained when a seed set selected according to the model on the column is assessed with the model on the row. For example, the cell (ICM, LTM) indicates that a seed set selected via the greedy algorithm for the LTM produced influence spread that was 99.8\% optimal when diffusion actually occurred under the ICM. Given that all entries are greater than 90\%, determining model fit by examining the magnitude of influence spread achieved under its seeding strategy leads to great ambiguity. Since all of these models result in high influence spread, \emph{any} model could account for the ``true'' underlying process of diffusion. 

\begin{figure}[ht]
\caption{}
    \begin{subfigure}[b]{\columnwidth}
        \centering
			\fontsize{8.0pt}{10.0pt}
		\selectfont
		\begin{tabular}{lccc}
			\hline
			& \multicolumn{1}{l}{ICM} & \multicolumn{1}{l}{LTM} & \multicolumn{1}{l}{AJM} \\ \hline
			ICM & 100                     & 99.8                    & 98.6                    \\
			LTM & 99.6                    & 100                     & 98.8                    \\
			AJM & 97.4                    & 96.1                    & 100                     \\ \hline
		\end{tabular}

			\begin{tabular}{lccc}
		\hline
		& \multicolumn{1}{l}{ICM} & \multicolumn{1}{l}{LTM} & \multicolumn{1}{l}{AJM} \\ \hline
		ICM & 100                     & 98.4                    & 99.8                    \\
		LTM & 99.9                    & 100                     & 98.9                    \\
		AJM & 93.7                 & 97.8                   & 100                     \\ \hline
	\end{tabular}

			\begin{tabular}{lccc}
	\hline
	& \multicolumn{1}{l}{ICM} & \multicolumn{1}{l}{LTM} & \multicolumn{1}{l}{AJM} \\ \hline
	ICM & 100                     & 98.7                    & 99.4                    \\
	LTM & 99.3                    & 100                     & 99.8                    \\
	AJM & 96.3                 & 93.9                  & 100                     \\ \hline
\end{tabular}
        \caption{Percentage of optimal influence spread achieved when network is seeded according to the column model and evaluated according to the row model. Networks top to bottom: Homeless, India, SBM.}
        \label{fig:cross-validation}
    \end{subfigure}%
    \\
    \begin{subfigure}[b]{\columnwidth}
        \centering
        \includegraphics[height=1.05in]{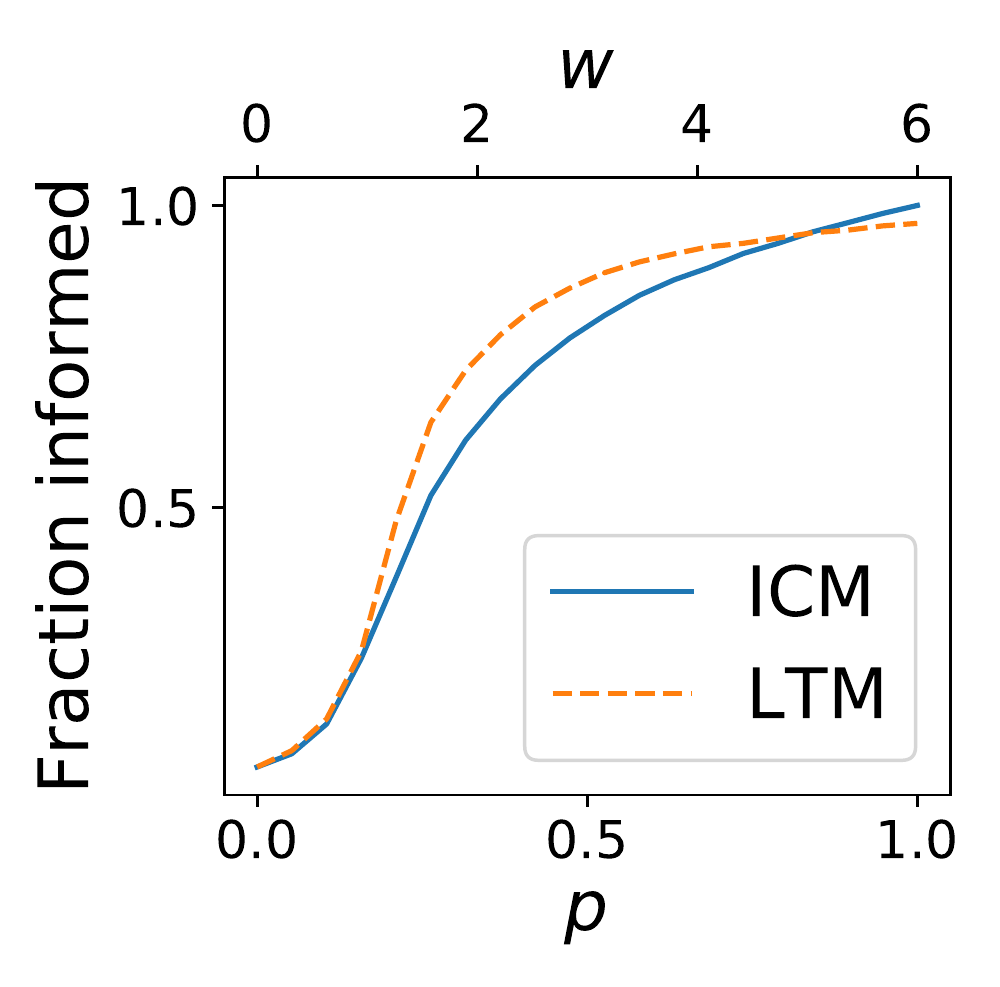}
        \includegraphics[height=1.05in]{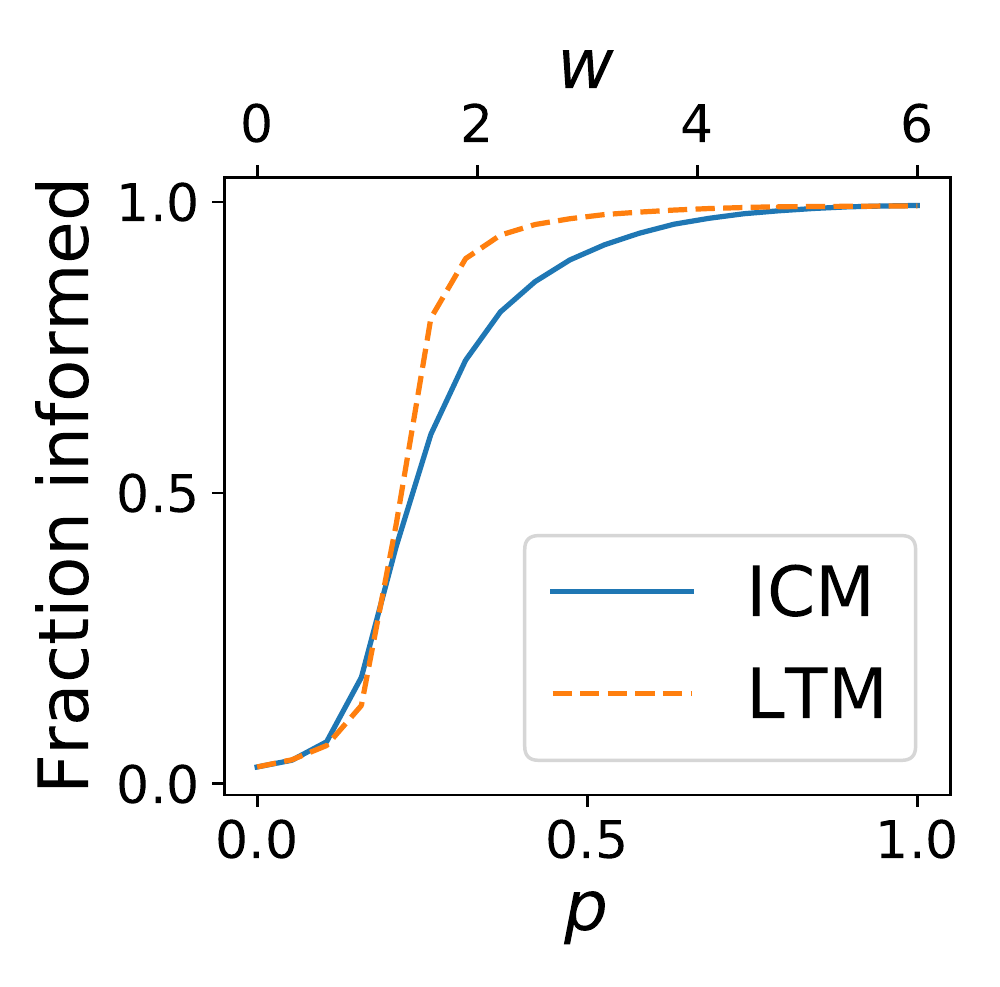}
        \includegraphics[height=1.05in]{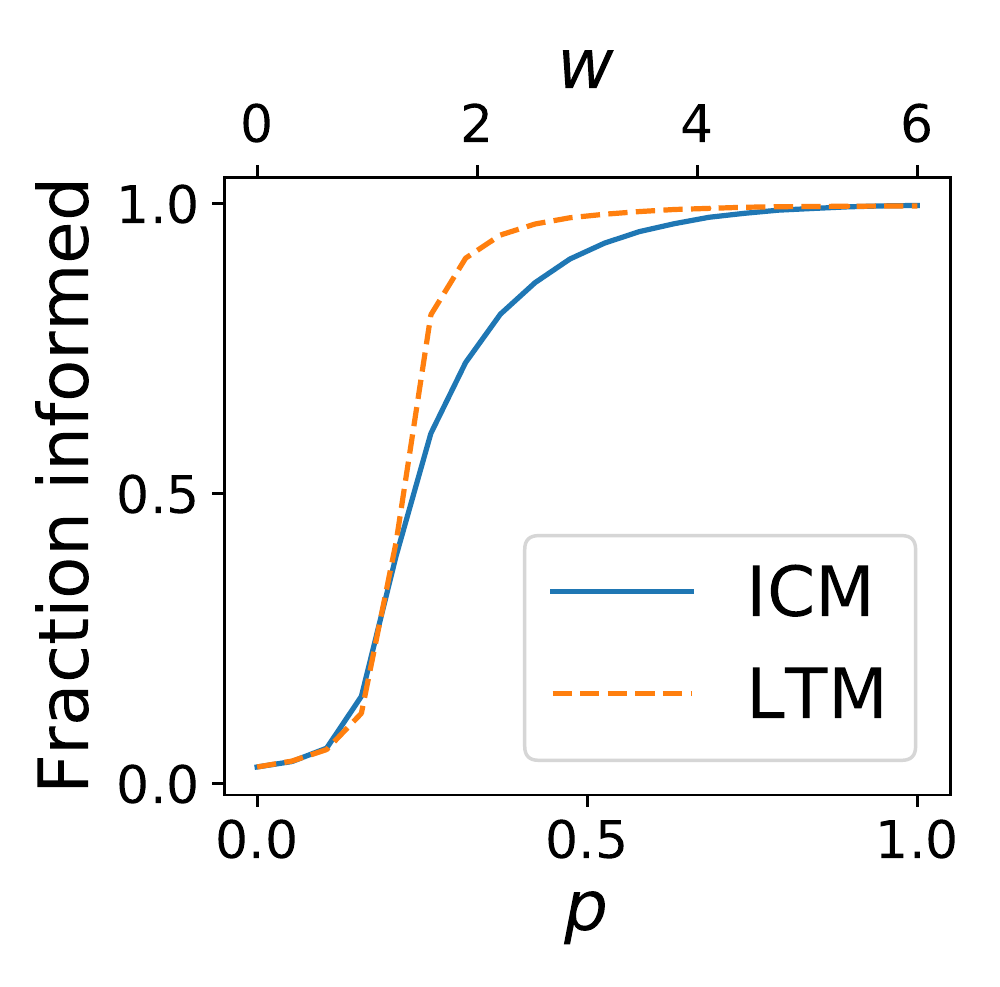}
        \caption{Fraction of all nodes informed under ICM and LTM with varying parameters. Networks left to right: Homeless, India, SBM.}
        \label{fig:parameter}
    \end{subfigure}
\end{figure}

One explanation for this phenomenon points to the community structure common in social networks. Algorithms for influence maximization under the ICM tend to distribute seed nodes across different communities to avoid the redundancy of seeding the same community multiple times. But seeding according to the AJM results in a similar recommendation to ensure diversity among seed nodes. Hence, attaining high influence spread is  insufficient for identifying a model as the true diffusion mechanism. On the one hand, achieving comparable final magnitudes of influence spread is a handy property for influence maximization tasks, as it suggests that high-quality results are attainable even when the true model is uncertain. However, many important influence maximization tasks require a descriptively accurate diffusion model, not just one that works by coincidence. 

Next, in Figure \ref{fig:parameter}, we show that common diffusion models are capable of reproducing \emph{any} observed level of total influence spread. Each plot gives the fraction of the graph influenced by a random set of 10 seed nodes under the ICM and LTM as we vary a parameter for each model. For the ICM, we vary the propagation probability $p$. For the LTM, we assign each edge $(u,v)$ a weight $\frac{w}{deg(v)}$ and vary $w$. Each line is an average of 30 draws of the random seeds. Under both models, any level of influence can be explained by a parameter choice in either of the models. We conclude that even if a given model exactly replicates the observed amount of influence spread in a network, this provides \emph{no evidence} that the model truly describes the underlying diffusion process. Hence, we must use a finer-grained assessment such as node-level activations to produce accurate diffusion models.

\textbf{Toward Node-Level Influence Spread } The prevailing methodology's blind spot to node-level information spread also engenders severe limitations in the actual deployment of the influence maximization problem in the real world. Seeding networks to maximize the scalar volume of influence spread is unproblematic when one is agnostic about \textit{who} is influenced. 
But in many sustainability domains, influence maximization in the field is not a crude game of coverage, and the individual identities or attributes of influenced nodes carry import. In the Introduction, we referenced the variety of network interventions that seek to target particular individuals or sub-populations. Here, a single metric capturing the magnitude of influence spread achieved is insufficient in determining the success of a network intervention. Thus even when in cases when seeding according to the AJM rather than the ICM or LTM would result in similar levels of influence spread, social programs would prefer seeding according to the model that makes accurate node-level predictions as it will also be better equipped to target specific individuals. We have shown that in these cases, the AJM significantly outperforms leading models.

\section{Discussion}
A breadth of research has investigated information diffusion on online networks, but the problem of influence maximization remains under-explored on natural networks. By analyzing node-level data from the first ever large-scale study of influence maximization on physical social networkw, we show that neither of the prevailing models of information diffusion--the Independent Cascade and Linear Threshold--could account for the empirical findings. Even after fitting the best of these models to the data, they \textit{perform worse than a random classifier} in predicting a node's influence status.    



We approached the shortcomings of the dominant models with an open mind to related research that may inform our understanding of information diffusion in this domain. Our proposed Activation Jump Model (AJM) draws from a lineage of work in multi-agent systems and social network theory that suggests that 1) social exchange need not only occur along network ties and 2) an individual's influencing behavior is affected by her surrounding community. Of particular note, we model seed set structural diversity as conferring benefits to each node's influencing level. 

The AJM is a more inclusive model of diffusion and superior to leading models in its predictive prowess. When validated on three real-world networks with information diffusion data, the AJM issues predictions of node-level influence spread that \textit{improve upon the best ICM and LTM predictions by 60\% to 110\%}. Moreover, as the AJM is submodular and non-monotone, we adopt a seeding algorithm that achieves a $\frac{1}{4}$-approximation to the optimal influence spread. Thus high-efficacy influence maximization under the AJM is computationally ready to be deployed in the real world. 

It has long been accepted in the social sciences that the link between individual and group social behaviors is bidirectional \cite{mead1934mind}. A multi-agent team perspective is thus particularly suited to describe peer-to-peer information diffusion in the natural world, where a group's social dynamic impacts how individual members will behave in spreading information. By modeling team-formation, a central component of many network interventions, the AJM significantly updates the influence maximization problem for natural world settings. Its framework, with flexible activation level and landing distribution functional forms, also allows for contextually relevant information such as node-specific attributes like gender and ethnicity to be incorporated when deployed in real-world network applications.

\bibliographystyle{aaai}

\bibliography{AJM-bibliography}
\end{document}